\documentclass[
    onecolumn,aps,pra,
    superscriptaddress,
    nofootinbib,notitlepage
]{revtex4-2}

\usepackage{graphicx} 
\usepackage[]{xcolor}

\usepackage{amssymb}
\usepackage{amsmath}
\usepackage{amsthm}
\usepackage{mathtools}

\usepackage{bm}
\usepackage{bbold}
\usepackage{siunitx}
\usepackage{physics}
\usepackage[version=4]{mhchem}
\usepackage{enumitem}
\usepackage{tabularx}

\usepackage{qcircuit}

\usepackage{algpseudocode}
\usepackage[ruled,vlined]{algorithm2e}
\RestyleAlgo{ruled}

\usepackage[colorlinks=true, linkcolor=niceblue, citecolor=gray]{hyperref}

\definecolor{niceblue}{rgb}{.1, .25, .8}

\makeatletter
\def\@endtheorem{\endtrivlist}
\makeatother

\newtheorem{theorem}{Theorem}
\newtheorem{definition}[theorem]{Definition}
\newtheorem{lemma}[theorem]{Lemma}
\newtheorem{proposition}[theorem]{Proposition}

\newtheorem{corollary}[theorem]{Corollary}

\DeclareMathOperator{\E}{\mathbb{E}}

\newcommand{\torder}[1]{\widetilde{\mathcal{O}}\left( #1 \right)}
\newcommand{\Paragraph}[1]{\noindent \textbf{#1:}}

\newtheorem{remark}[theorem]{Remark}

\newtheorem*{question*}{Question}

\newcommand{\one}{\mathbb{1}}

\DeclareMathOperator{\poly}{poly}
\DeclareMathOperator{\rt}{rt}
\DeclareMathOperator{\polylog}{polylog}

\newcommand{\later}[1]{}

\DeclareMathOperator{\negl}{negl}

\newcommand{\SWAP}{\mathsf{SWAP}}

\newcommand{\newreptheorem}[2]{\newtheorem*{rep@#1}{\rep@title}\newenvironment{rep#1}[1]{\def\rep@title{#2 \ref*{##1}}\begin{rep@#1}}{\end{rep@#1}}}
\makeatother

\newcommand{\swap}{\mathsf{SWAP}}

\newcommand{\PP}{\mathbb{P}}
\newcommand{\EE}{\mathbb{E}}
\newcommand{\TV}{\mathrm{TV}}

\usepackage{subcaption}
\usepackage{comment}

\begin{document}

\title{Arts \& crafts: Strong random unitaries and geometric locality}
\author{Marten Folkertsma}
\affiliation{QuSoft and UvA, Amsterdam, the Netherlands}
\author{Lorenzo Grevink}
\affiliation{QuSoft and CWI, Amsterdam, the Netherlands}
\author{Jonas Helsen}
\affiliation{QuSoft and CWI, Amsterdam, the Netherlands}
\author{Alicja Dutkiewicz}
\affiliation{QuSoft and CWI, Amsterdam, the Netherlands}

\begin{abstract}
We study the problem of constructing strong approximate unitary $k$-designs on $D$-dimensional grids (and more generally on Cartesian products of graphs), building on the work of Schuster et al.~\cite{schuster2025strong} which establishes strong unitary designs in 1D and in all-to-all connectivity. We provide two constructions. The first construction leverages the existing all-to-all connectivity result with general routing theory to provide flexible (but slightly suboptimal) strong $k$-designs in arbitrary connectivities. The second construction is more direct, requires no auxiliaries and has provably optimal depth (in  the number of qubits $n$) for $D$-dimensional grids with constant dimension. Combining these techniques also allows us to construct strong pseudorandom unitaries on $D$-dimensional grids with provably optimal depth.
\end{abstract}

\maketitle

\section{Introduction}
Random quantum processes have been a growing field of study in quantum computing and quantum information theory, having significant applications in numerous fields of study including quantum machine learning~\cite{larocca2025barren}, many-body physics~\cite{fisher2023random}, and quantum gravity~\cite{hayden2007black, brown2023quantum}.
Moreover, they are widely used in practical applications such as device benchmarking~\cite{emerson2005scalable, elben2023randomized}, state tomography~\cite{huang2020predicting, helsen2023thrifty} and quantum cryptography~\cite{ji2018pseudorandom, kretschmer2021quantum}.
A natural theoretical model for randomness in quantum systems is given by Haar-random unitaries, which correspond to the uniform distribution over the unitary group.
However, simple counting arguments show that implementing a Haar-random unitary requires resources exponential in the system size~\cite{knill1995approximation}, making them computationally intractable and physically unrealistic.

Consequently, a substantial line of work~\cite{harrow2009random, nakata2016efficient,hunter2019unitary, haferkamp2022random,haferkamp2023efficient,harrow2023approximate, metger2024simple, schuster2025random, chen2025incompressibility, ma2025construct, schuster2025strong, haah2025efficient} is devoted to finding ensembles of unitaries that are efficient to construct, but appear to be Haar random in certain contexts.
For example, an ensemble that appears to be random to any computationally-bounded adversary is called a \emph{pseudorandom unitary} (PRU), and an ensemble that appears to be random to any adversary that samples a unitary at most $k$ times is called a \emph{unitary $k$-design}.
More recently, \emph{strong} variants of these notions have been introduced, in which the adversary is allowed access not only to a unitary $U$, but also to its inverse $U^\dagger$~\cite{schuster2025strong}.
These provide a more faithful approximation to Haar randomness, as many quantum algorithms and physical processes naturally involve both forward and inverse evolutions, whereas the standard (weak) design condition provides no guarantee in this setting.
Therefore, in this work we focus on strong designs and PRUs.

Recent works have provided efficient constructions of such random unitary ensembles via random quantum circuits -- sequences of local random gates, typically acting on pairs of qubits.
A crucial aspect of these constructions determining the scaling of resources with the number of qubits $n$ is the underlying connectivity, which specifies which pairs of qubits can interact.
In the most powerful setting of \emph{all-to-all} connectivity, strong designs can be constructed in depth $\order{\log(n)}$~\cite{schuster2025strong}, and weak designs can be achieved in a surprisingly low depth of $\order{\log\log(n)}$~\cite{cui2025unitary}.
At the other extreme, when qubits are arranged in a one-dimensional chain and can only interact with their nearest neighbors, a lightcone argument
shows that strong designs require depth $\Omega(n)$, and this bound is achieved by~\cite{schuster2025strong}, while weak designs can be achieved in
depth $\order{\log n}$~\cite{schuster2025random, laracuente2026approximate}, representing an exponential slowdown in both cases compared to all-to-all.
These modern, optimal constructions rely on \emph{gluing lemmas}, which allow combining small designs into larger ones.
More general $D$-dimensional connectivities were studied prior to the introduction of strong designs and gluing lemmas, achieving the depth $d = \order{n^{1/D}}$~\cite{harrow2023approximate} for weak designs, which turned out to be not at all tight.
The optimal depth for strong designs in intermediate connectivities remains an open question.

This question is not merely theoretical -- intermediate connectivity structures are highly relevant in practice.
While some emerging platforms based on neutral atoms~\cite{bluvstein2026fault} or ion traps~\cite{ransford2025helios} are expected to achieve (nearly) all-to-all connectivity, many others are designed to respect local geometries.
For example, superconducting qubits~\cite{alam2026onset, bravyi2021hadamard, google2025quantum} and silicon quantum dots \cite{li2018crossbar} are typically laid out in regular two-dimensional grids.
Higher dimensional geometries also arise: for instance, irregular three-dimensional geometries naturally appear in NV centers platforms, and three-dimensional lattices error-correcting codes have been studied due to the possibility of transversal magic gates~\cite{bombin2015gauge}.
Moreover, beyond hardware considerations, regular lattices in arbitrary dimensions have been subject of theoretical studies of quantum chaos~\cite{liao2022effective}.
Together, these considerations motivate the main question of this paper:
\begin{question*}
    At what depth do random circuits form strong unitary designs and PRUs in general connectivities?
\end{question*}
In particular we focus on $D$-dimensional grids, which provide a natural and analytically tractable model of geometrically local quantum systems.
Such grids interpolate between one-dimensional chains and all-to-all architectures.
Building on the results of Ref.~\cite{schuster2025random}, we show that strong designs and PRUs can be constructed in a circuit depth of $\Theta(n^{1/D})$, matching the natural lightcone scaling in $D$-dimensional grids.

\subsection{Summary of results}

We address the question using two different approaches.
Our main result builds on the \emph{strong gluing lemma} of Schuster et al.~\cite{schuster2025strong}, which shows that small strong unitary designs can be composed into larger strong designs, provided they are sandwiched between large strong 2-designs.
We combine this with a new construction of strong 2-designs on $D$-dimensional grids to obtain optimal-depth strong unitary designs without auxiliary qubits.
\begin{theorem}[Theorem~\ref{theorem:gluing_lattice_result}]\label{theorem:main_better}
    For any constant $D$, there exist strong $\varepsilon$-approximate unitary $k$-designs on a $D$-dimensional grid in depth $d = \order{n^{1/D} + k\log^7(k)\log(nk/\varepsilon)}$, without any auxiliary qubits.
\end{theorem}
Note that our result also holds for $D = 1$, providing a small improvement over the results in Schuster et al.\ in some regimes. 

We provide optimal-depth strong pseudorandom unitaries using the same gluing framework, under the same well-established cryptographic assumption of sub-exponential hardness of learning with errors (LWE) as in Ref.~\cite{schuster2025strong}.
\begin{theorem}[Theorem~\ref{theorem:pru}]\label{theorem:main_pru}
    Assuming sub-exponential hardness of LWE, for any constant $D$, there exists a strong PRU on a $D$-dimensional grid in circuit depth $d = \Theta(n^{1/D})$, without using auxiliary qubits.
\end{theorem}
These constructions are optimal in the number of qubits $n$ for constant $D$.
\begin{remark}
    Both results match the lower bound, with respect to $n$, provided by Proposition~\ref{prop:lower_bound} when $D$ is a constant.
\end{remark}

Our alternative approach is straightforward, but nevertheless comes surprisingly close to the lower bounds provided by the lightcone arguments. The standard approach to implementing non-local interactions in a device constrained by limited connectivity is to intersperse a layer of non-local interactions with permutations, bringing qubits that need to interact close to one another. The task of implementing these permutations via neighboring operations is known as \emph{routing}. There is a strong line of research on finding optimal routing protocols respecting several interaction topologies using heuristics~\cite{Maslov_2008, Shafaei_2013, Shafaei_2014, Kole_2016, Booth_2018, Zulehner_2018, Siraichi_2019, Cowtan_2019, Li_2019, Wagner_2023}, as well as more explicit results on circuit overheads when implementing these routing protocols \cite{Kivlichan_2018, Childs_2019, Ogorman_2019, Yuan_2025}. 

We directly apply results from routing to the random unitary designs given by Schuster et al.~\cite{schuster2025strong}, which provides random unitary designs of near optimal depth. 

\begin{theorem}[Corollary~\ref{cor:routed_result}]\label{theorem:main_designs}
For any constant $D$, there exist strong $\varepsilon$-approximate unitary $k$-designs in the following circuit depth using circuits on a $D$-dimensional grid:
\begin{enumerate}
        \item $d = \order{(nk)^{1/D}\poly\log(n,k) \cdot \log\log(1/\varepsilon)}$ with $\torder{nk}$ auxiliary qubits.
        \item $d = \order{n^{1/D} k \cdot \poly\log(n) \cdot \log\log(k/\varepsilon)}$ with $\torder{n}$ auxiliary qubits.
\end{enumerate}
\end{theorem}

\noindent Theorem~\ref{theorem:main_designs} gives a better scaling in $k$ and $\varepsilon$ that Theorem~\ref{theorem:main_better}, at the cost of using auxiliary qubits. \\

Our results extend to connectivities of the form of a Cartesian product of graphs, which includes $D$-dimensional grids as a special case.
However, many practically relevant connectivities do not fall into this family -- for instance, the hexagonal lattices used in some superconducting qubit platforms~\cite{alam2026onset} or the biplanar graphs underlying certain error correcting codes~\cite{bravyi2021hadamard}.
Extending our techniques to such connectivities remains an interesting open problem.
To this end we note that the routing strategy that yields Theorem \ref{theorem:main_designs} is structurally very similar to the construction of the strong two design underlying Theorem \ref{theorem:main_better}.
This suggest that routing results for more general graphs could perhaps be more directly transformed into constructions of strong two-designs on more general connectivities. 

\newpage

\section{Background}

\noindent In this section we review the necessary background for our results.

\subsection{Preliminaries}

\Paragraph{Connectivity graph} The connectivity of a quantum device can be specified by a \textit{connectivity graph}. The vertices of the graph represent the qubits, while an edge indicates that one can apply an entangling gate on the pair of qubits that it connects. In this work we study how to construct designs on a connectivity graph that is a Cartesian product of smaller graphs. A Cartesian product of graphs is defined as follows.

\begin{definition}
\label{def:cartasian_product}
    Let $G = (V,E), G' = (V', E')$ be two graphs. The Cartesian product of these graphs $G\cross G'$ is defined as the graph with vertex set $V \cross V' = \{(v,v')|v\in V, v' \in V'\}$ and with an edge $(u,u') \sim (v,v')$ if either $u = v \land (u',v') \in E'$ or $u'=v' \land (u,v) \in E$.
\end{definition}

\noindent A $D$-dimensional grid $G_{D,\mathrm{grid}}$ is defined as the Cartesian product of $D$ line graphs of $n^{1/D}$ qubits each, adding up to a total of $n$ qubits, i.e.,
\begin{equation}\label{eqn:grid_prod_linegraphs}
    G_{D,\mathrm{grid}} = L_{n^{1/D}}^{0}\cross L_{n^{1/D}}^{1}\cross \dots \cross L_{n^{1/D}}^{D-1},
\end{equation}
where the superscript is only added to keep track of the different line graphs.
Our work focuses on $D$-dimensional grids, but our results can easily be generalized to Cartesian products of other graphs.\\

\Paragraph{TV distance} We define the total variation (TV) distance between two probability distributions 
\begin{definition}[TV distance]
    Let $P$ and $Q$ be two probability distributions over events $\Omega$, then the TV distance is
    \begin{equation}
        \TV(P,Q) = \frac12\sum_{x\in \Omega} \left|P(x) - Q(x)\right|_1.
    \end{equation}
\end{definition}
\noindent A useful fact about the TV distance is that it respects subadditivity.
\begin{lemma}[Subadditivity TV distance]
\label{lem:subadditivity TV distance}
Given two product distributions $P = P_1 \times P_2$ and $Q = Q_1 \times Q_2$ over $\Omega =\Omega_1 \times \Omega_2$, then the TV distance between $P$ and $Q$ can be upper bounded by the local TV distances as
\begin{equation}
    \TV(P,Q) \leq \TV(P_1, Q_1) + \TV(P_2, Q_2).
\end{equation}
\end{lemma}

\Paragraph{Strong unitary designs and PRUs}  We introduce some notions of ensembles over the unitary group that resemble the Haar measure.

A strong unitary $k$-design $\mathcal{E}$ is an ensemble that is indistinguishable from Haar-random in any experiment that samples $U, U^T, \bar{U}$ or $U^\dagger$ at most $k$ times up to some $\varepsilon$. More precisely, if $A$ is an experiment that samples from $\mathcal{E}$ at most $k$ times, than the expected output of the algorithm is $\varepsilon$-close in trace distance to the outcome if $A$ would have sampled from the Haar measure. An algorithm that samples a unitary $k$ times outputs a state
\begin{equation}
        \ket{\mathcal{A}^U_k} = \prod_{0 < i \leq k} (W_i (V_i\otimes \one)) \cdot W_0\ket{0^{n + m}}
    \end{equation}
for some unitaries $W_i$, where $V_i \in \{U, U^T, \bar{U}, U^\dagger\}$ depending on $i$. Here, $m$ is the number of auxiliary qubits and can be arbitrarily large. We thus say that $\mathcal{E}$ is an $\varepsilon$-approximate strong unitary $k$-design if
\begin{equation}
    \left|\E_{U\sim \nu} \Tr[O\ketbra{\mathcal{A}^U_k}] - \E_{U\sim H} \Tr[O\ketbra{\mathcal{A}^U_k}]\right| < \varepsilon
\end{equation}
for every choices of $m, W_i, V_i$ and any observable $O$.
A strong unitary $k$-design can easily be extended to also being secure against controlled queries by a reduction laid out in~\cite{sheridan2009approximating}. A strong design with $\varepsilon = 0$ is called an exact design.
Strong unitary designs stand in contrast with \textit{weak unitary designs}, that are only secure against adversaries that query the unitary $U$, but not $U^T, \bar{U}$ or $U^\dagger$.

A similar notion to strong unitary designs are strong pseudorandom unitaries (PRU). Strong PRUs are ensembles over the unitary group that are indistinguishable from Haar random to any experiment that runs in polynomial time. In addition we require that PRUs are efficiently computable. More precisely, an ensemble $\mathcal{E}$ over the unitary group is a strong PRU if
    \begin{itemize}
        \item for every experiment $\mathcal{A}^U$ that runs in polynomial time and has oracle access to the unitaries $U$, $U^T, \bar{U}$ and $U^\dagger$, and for any two-outcome observable $O$, we have that
        \begin{equation}
            |\E_{U\sim \mathcal{E}} \Tr[O\ketbra{\mathcal{A}^U}] - \E_{U\sim H} \Tr[O\ketbra{\mathcal{A}^U}]| < \negl(n),
        \end{equation}
        where $\negl(n)$ is any function that is smaller than every inverse polynomial and
        \item there is a polynomial-time algorithm that implements $\mathcal{E}$.
    \end{itemize}    

\noindent As PRUs are secure against experiments that query the unitary a polynomial number of times, one might think that PRUs are always unitary $k$-designs for fixed $k$. Note, however, that there is no restriction on the runtime of a distinguishing experiment for unitary designs, so PRUs and unitary designs are incomparable.
On the other hand, an ensemble that is a strong unitary $k$-design for $k$ larger than any polynomial, is automatically a strong PRU.
Similarly as for unitary designs, one can define weak PRUs and strong PRUs, depending on whether they are secure against inverses.

It was proven in~\cite{ma2025construct} that the existence of a quantum-secure one-way functions implies the existence of strong PRUs. Later,~\cite{schuster2025strong} proved that strong PRUs can be constructed in logarithmic depth in the all-to-all connectivity, relying on the stronger assumption of sub-exponential hardness of learning with errors (LWE).

\subsection{Summary of Ref.~\cite{schuster2025strong}}

\noindent In this section we review the results of Schuster et al.~\cite{schuster2025random} that we will use in this work.
For completeness, we provide proofs for several statements that are implicit in Ref.~\cite{schuster2025random}.\\

\begin{figure}
\centering
\begin{subfigure}[t]{0.45\textwidth}
    \includegraphics[height=1.2in]{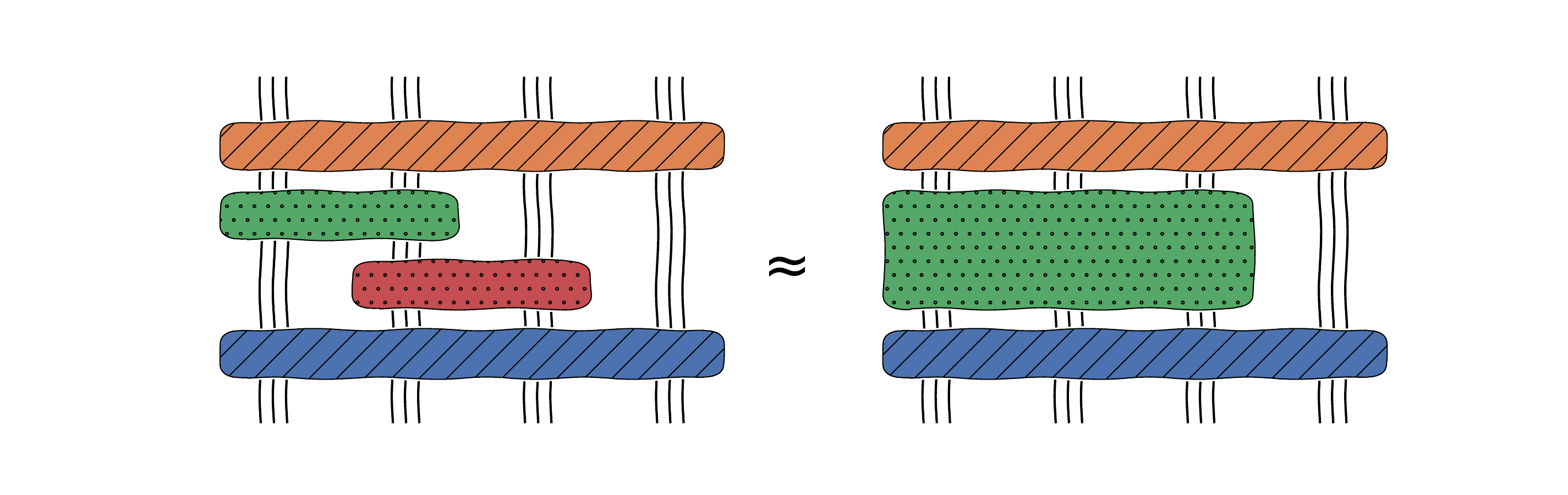}
    \caption{A key step in our proofs is the strong gluing Lemma~\ref{lemma:strong_gluing}. Two strong $k$-designs can be glued together, provided that on both sides a strong 2-design is applied. Note that the strong 2-designs also act on an additional register not touched by the $k$-designs.}
    \label{fig:glueing_illustration}
\end{subfigure}
\begin{subfigure}[t]{0.45\textwidth}
    \includegraphics[height=1.2in]{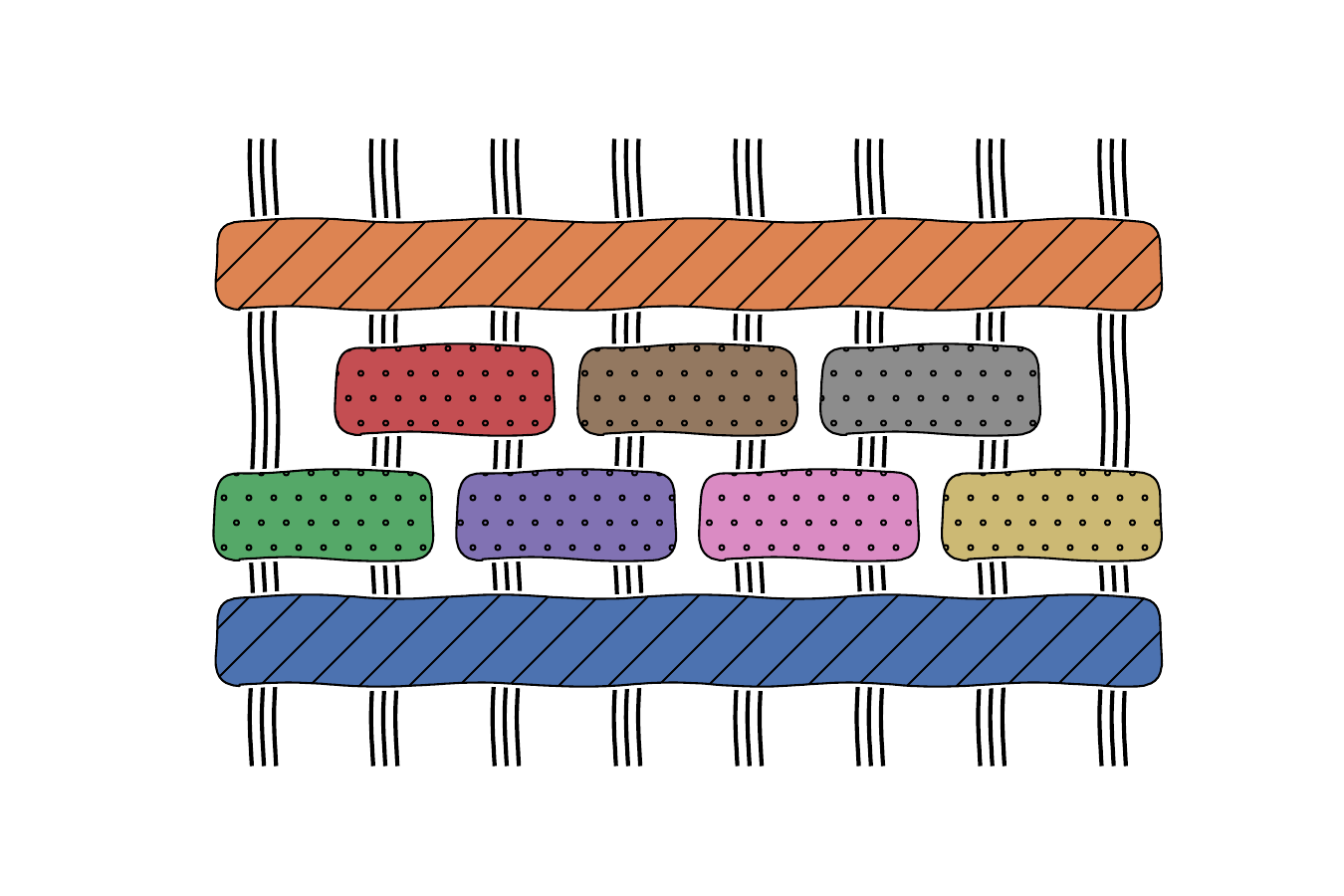}
    \caption{The strong gluing lemma can be applied iteratively as described in Corollary~\ref{corollary:gluing_lemma}. We first apply a strong 2-design on all qubits. Then we apply two layers of strong $k$-designs on patches of $2\xi$ qubits. We end with another strong 2-design on all qubits. The procedure results in a strong $k$-design on all qubits.}
    \label{fig:gluing_iteratively}
\end{subfigure}
\end{figure}

\Paragraph{Gluing lemma for designs}
Our results build on the strong gluing lemma, which allows combination of two unitary designs with sufficient overlap into one larger unitary design, provided access to strong unitary $2$-designs, as illustrated in Fig.~\ref{fig:glueing_illustration}. 
\begin{lemma}[Gluing strong random unitaries, Lemma~1 in~\cite{schuster2025strong}\footnote{
Schuster et al.\ state the lemma without the additional register $d$ which is not touched by the gluing process, while this additional register is essential in their iterative arguments that result in their main results. Note that Lemma~\ref{lemma:strong_gluing} does not trivially follow from the statement without this additional register stated in Lemma~1 of~\cite{schuster2025strong}. That is because one cannot just ``pull out'' a 2-design out of another 2-design, in a context in which an adversary can sample the unitary more than 2 times. The proof that Schuster et al.\ provide for their lemma, however, does include an additional register $d$, so Lemma~\ref{lemma:strong_gluing}, and therefore all results in~\cite{schuster2025strong}, still hold.
}]
\label{lemma:strong_gluing}
Given four subsystems, $a,b,c,d$ where $a,b,c$ are at least size $\xi$, consider the unitary ensemble $U_1 = D_{abcd}U_{ab}U_{bc}C_{abcd}$, where $C_{abcd}$ and $D_{abcd}$ are strong $\varepsilon_2$-approximate unitary $2$-designs and $U_{ab}$, $U_{bc}$ are strong $\varepsilon_{ab}$- and $\varepsilon_{bc}$-approximate unitary $k$-designs on their respective subsystems. Consider another ensemble $U_2 = D_{abcd}U_{abc}C_{abcd}$, where $U_{abc}$ is Haar random. Now any experiment that samples a unitary, its inverse, transverse or conjugate $k$ times, cannot distinguish $U_1$ from $U_2$ up to 
$\varepsilon = \varepsilon_{ab} + \varepsilon_{bc} + \order{k^2/2^{(3/16)\xi}}+ \order{k^{5/8}\varepsilon_2^{1/8}}$ in trace distance.
\end{lemma}
One can iteratively apply Lemma~\ref{lemma:strong_gluing} to construct strong $k$-designs on larger systems from smaller building blocks, resulting in the following procedure illustrated in Fig.~\ref{fig:gluing_iteratively}.
\begin{corollary}[Theorem 5 from \cite{schuster2025strong}]\label{corollary:gluing_lemma}
    We divide our $n$ qubits into patches of size $\xi$, where we assume that $m = n/\xi$ is an integer. We define the following protocol:
    \begin{enumerate}
        \item We apply an $\varepsilon/3$-approximate strong 2-design on all qubits
        \item We apply an $\varepsilon/(3n)$-approximate strong $k$-designs on the patches $(2i-1, 2i)$ of size $2\xi$, for every $1\leq i \leq m/2$
        \item We apply an $\varepsilon/(3n)$-approximate strong $k$-designs on the patches $(2i, 2i+1)$ of size $2\xi$, for every $1\leq i \leq m/2$
        \item We apply an $\varepsilon/3$-approximate strong 2-design on all qubits
    \end{enumerate}
    The ensemble resulting from this protocol forms a strong $\varepsilon$-approximate unitary $k$-design when $\xi \geq \frac{16}{3} \log_2(nk^2/\varepsilon) + \order{1}$.
\end{corollary}
To apply Corollary~\ref{corollary:gluing_lemma}, it remains to construct the small strong $k$-designs that are combined by the gluing procedure. This is achieved by the following lemma.
\begin{lemma}[Lemma 3 in \cite{schuster2025strong}]
\label{lem:1d-strong-kdesigns}
    1D random circuits on $n$ qubits form strong $\varepsilon$-approximate unitary $k$-designs in depth
    $d = \order{\log^7(k) \left(nk + \log(1/\varepsilon)\right)}$
    for any $\varepsilon \geq 2k^2/2^n$.
\end{lemma}
Combining Corollary~\ref{corollary:gluing_lemma} and Lemma~\ref{lem:1d-strong-kdesigns} leads to the following theorem.
\begin{theorem}[Theorem 1 in~\cite{schuster2025strong}]
\label{thm:all-to-all-designs}
    Strong $\varepsilon$-approximate unitary $k$-designs can be realized in the following circuit depths:
    \begin{enumerate}
        \item $d = \order{\log n + \log k \cdot \log\log(nk/\varepsilon)}$ using all-to-all circuits with $\torder{nk}$ auxiliary qubits.
        \item $d = \order{\log n + k \cdot \log\log(nk/\varepsilon)}$ using all-to-all circuits with $\torder{n}$ auxiliary qubits.
        \item $d = \order{k \cdot \operatorname{poly}\log k \cdot \log(n/\varepsilon) + \log n \cdot \log(n/\varepsilon)}$ for all-to-all random circuits of independent Haar-random two-qubit gates without auxiliary qubits.
    \end{enumerate}
\end{theorem}

\Paragraph{Gluing lemma for PRUs} 
An analogous gluing procedure to Corollary~\ref{corollary:gluing_lemma} can be used to construct strong PRUs.
However, one needs to be careful to choose the parameters $\varepsilon$ and $\xi$ such that the ensemble is secure against polynomially bounded adversaries.
We summarize this in the following corollary.
\begin{corollary}[Theorem 6 in \cite{schuster2025strong}]\label{corollary:PRU_gluing}
    Consider the procedure from Corollary~\ref{corollary:gluing_lemma}. For Steps 1 and 4, we choose $\varepsilon$ smaller than any polynomial in $n$, and in Steps 2 and 3 we replace the strong designs by strong PRUs with $\poly(n)$-time security on $2\xi$ qubits. If we choose $\xi = \order{\log^{1+\delta}(n)}$ for any fixed $\delta>0$, then the procedure results in a strong~PRU.
\end{corollary}
\begin{proof}
    From the perspective of any polynomially-bounded adversary, the $2\xi$-sized strong PRUs are indistinguishable from Haar random. In particular, they are indistinguishable from $k$-designs for any $k$. We use Corollary~\ref{corollary:gluing_lemma} to see that, if we choose $\xi = \log^{1 + \delta}(n)$, that then the ensemble results in a strong $k$-design for $k = \order{2^{\log^{1+\delta}(n)}}$ and $1/\varepsilon = \order{2^{\log^{1+\delta}(n)}}$. These functions grow faster than any polynomial, so we proved the corollary.
\end{proof}
To complete the construction, we again need building blocks for small strong PRUs, which can then be glued together. These are provided by the following lemma.
\begin{lemma}[Appendix E.3 in~\cite{schuster2025strong}]\label{lemma:PRU_buildingblocks}
    Assuming the sub-exponential post-quantum hardness of LWE, there exist sub-exponentially secure, auxiliary-free strong PRUs computable in depth $\poly(n)$ with all-to-all circuits, without auxiliary qubits.
\end{lemma}
Note that the PRUs are sub-exponentially secure, meaning there exists some constant $\delta$ such that the PRU is secure against experiments that run in time~$2^{\order{n^\delta}}$.
This stronger notion of security is crucial for the application of the strong gluing lemma.
\\

\Paragraph{Lower bounds}
Schuster et al.~also provide lower bounds on the depth required to implement strong designs on one-dimensional and all-to-all connectivity.
These bounds naturally extend to the $D$-dimensional grids as well.
\begin{proposition}[Implicit from Proposition~1 in~\cite{schuster2025strong}]
\label{prop:lower_bound}
Any circuit ensemble on $n$ qubits on a $D$-dimensional grid that forms a strong $\varepsilon$-approximate unitary $2$-design for $\varepsilon < 1/4$ requires depth $d = \Omega(n^{1/D})$. The same lower bound holds for strong PRUs.
\end{proposition}

\begin{proof}
Consider the experiment measuring $M = \sum_{\mathbf{s}\in\{0,1\}^n} \frac{|\mathbf{s}|}{n}|\mathbf{s}\rangle\langle\mathbf{s}|$ on the state $U^\dagger Z_0 U |0^n\rangle$, where $Z_0$ is a single-qubit Pauli on the first qubit. For Haar-random $U$,
\begin{equation}
    \mathbb{E}_{U\sim H}\bigl[\langle 0^n|U^\dagger Z_0 U \cdot M \cdot U^\dagger Z_0 U|0^n\rangle\bigr] \geq \tfrac{1}{2},
\end{equation}
so any strong $\varepsilon$-approximate unitary $2$-design must achieve expectation value at least $\frac{1}{2} - \varepsilon$.

The operator $U^\dagger Z_0 U$ has support on at most $L$ qubits, where $L$ is the size of the lightcone of $U$ starting at the first qubit. A circuit $U$ with depth $d$ on a $D$-dimensional grid has a lightcone of at most $L \leq (2d)^D$ qubits. Therefore the fraction of qubits it can flip is at most $L/n$, giving
\begin{equation}
    \mathbb{E}_{U\sim\mathcal{E}}\bigl[\langle 0^n|U^\dagger Z_0 U \cdot M \cdot U^\dagger Z_0 U|0^n\rangle\bigr] \leq \frac{(2d)^D}{n}.
\end{equation}
For a strong $2$-design we need $(2d)^D/n \geq \frac12 - \varepsilon \geq \frac14$, giving $d = \Omega(n^{1/D})$. As the described experiment is efficient, the lower bound also holds for PRUs.
\end{proof}

\section{Optimal-depth strong designs}\label{sec:designs}

In this section we will construct an approximate strong unitary $k$-design in depth $d = \order{n^{1/D}}$.
The proof passes through the following steps.
First, in Lemma~\ref{lem:UUbar_design}, we construct a unitary design that is secure against any experiment that queries one copy of $U$ or $U^T$ and one copy of $U^{\dagger}$ or $U^*$.
Our construction relies on the underlying Cartesian product structure of the connectivity graph\footnote{Interestingly we find exactly the same depth overhead in these designs as is required in the routing protocol. The advantage of this approach is that the final circuit achieves the exact goal of being a design, instead of a subroutine that we can use to compile other circuits.} and is similar to the alternating construction introduced in~\cite{harrow2023approximate} (though our analysis applies to more general connectivity graphs).
Our second step, Theorem~\ref{theorem:2-designs_optimal}, is to note that the combining such a design with a weak design forms a strong design.
For the weak design, we use the low-depth $1$D construction of~\cite{schuster2025random}\footnote{This second step is likely not needed, but simplifies the proof (and anyway has subleading depth scaling).}.
Finally, in Theorem~\ref{theorem:gluing_lattice_result} we use the resulting strong $2$-design as a building block in the construction of Corollary~\ref{corollary:gluing_lemma}, obtaining strong $k$-design in optimal depth.

For the first step of our construction, we will need the following proposition from~\cite{schuster2025strong}.
\begin{proposition}[Proposition 2 in~\cite{schuster2025strong}]
\label{prop:U_barU}
Let $\mu$ be a unitary ensemble invariant under conjugation, transposition, and random single-qubit Pauli rotations at the input and output circuit layer. Then $\mu$ forms an approximate unitary $2$-design relative to any quantum experiment that queries one of $U$ or $U^T$ and one of $U^*$ or $U^\dagger$, with measurable error
\begin{equation}\label{eqn:TV-small-2-design}
    \varepsilon \leq \max_{Q\in \PP^*_n} \TV(\mathcal{E}_{\mu}(\cdot, Q), \mu_{\PP^*_n}),
\end{equation}
where $\mu_{\PP^*_n}$ is the uniform distribution over the set $\PP^*_n$ of non-identity Hermitian Pauli operators on $n$ qubits and $\mathcal{E}_{\nu}(\cdot, Q)$ is the distribution $\mathcal{E}_{\nu}(P,Q) = 2^{-2n} \EE_{U \sim \nu}|\tr(UQ U^\dagger P)|^2$.
\end{proposition}

We first show how to construct an ensemble that upper bounds Equation~\eqref{eqn:TV-small-2-design} on the Cartesian product of two graphs, provided that we already have ensembles on the respective graphs that are close to the uniform distribution.
\begin{lemma}[Pauli mixing under Cartesian product]
\label{lem:Cartasian_mixing}
     Let $\mu_R$, $\mu_C$ be two ensembles acting on $n_R$, $n_C$ qubits respectively, with circuit depths $d_R$ and $d_C$ over interaction graphs $G_R$ and $G_C$ respectively.
     Assume both are close to Pauli mixing:
    \begin{equation}
        \max_{Q\in \PP^*} \TV(\mathcal{E}_{\mu_b}(\cdot, Q), \mu_{\PP^*_{n_b}})\leq \varepsilon_b,
    \end{equation}
    for $b \in \{R, C\}$. Then there exists a circuit ensemble $\mu_g$ on $n = n_R \cdot n_C$ qubits, of depth $2d_R + d_C$ with respect to the product graph $G_g = G_R \cross G_C$ that satisfies
    \begin{equation}
    \begin{split}
        \max_{P\in \PP^*_{n}} \TV(\mathcal{E}_{\mu_g}(\cdot, P), \mu_{\PP^*_{n_R \cdot n_C}})
        &\leq (n_C+1) \varepsilon_R + n_C\left(\frac{1}{4} + \varepsilon_C\right)^{n_R/a} + \frac{4}{3}e^{-\frac{9n_R}{8}\left(1 - \frac{4}{3a}\right)^2} + \frac{n_C}{4^{n_R}}\\
        &
        \leq (n_C + 1)(\varepsilon_R + 2e^{-n_R \log(2)/3}) ,
    \end{split}
    \end{equation}
    for any $a > 4/3$, where the latter inequality holds when $\varepsilon_C \leq 1/4$ and $a = 3$. 
\end{lemma}
\begin{proof}
Throughout the proof, we view $G_g$ as a grid of $n_R \times n_C$ qubits. We refer to the copies of $G_R$ (the vertices $\{(r, c) : r \in V(G_R)\}$ for a fixed column index $c \in V(G_C)$) as \emph{rows}, and to copies of $G_C$ as \emph{columns}. There are $n_C$ rows and $n_R$ columns.
Any Pauli $P \in \PP_n$ can be decomposed either into rows or columns:
\begin{equation}
    P = \bigotimes_{c = 1}^{n_C}R_c = \bigotimes_{r = 1}^{n_R}Q_r
\end{equation}
with $R_c \in \PP_{n_R}$ acting on individual rows, and $Q_r \in \PP_{n_C}$ acting on individual columns.
We construct the ensemble $\mu_g = \mu_R^{\otimes n_C} * \mu_C^{\otimes n_R} * \mu_R^{\otimes n_C}$ in three layers: 
\begin{enumerate}
    \item Sample $n_C$ unitaries from $\mu_R$ and apply them in parallel to each row.
    \item Sample $n_R$ unitaries from $\mu_C$ and apply them in parallel to each column.
    \item Sample $n_C$ unitaries from $\mu_R$ and apply them in parallel to each row.
\end{enumerate}
By construction, the circuit depth with respect to $G_g$ is $2d_R + d_C$.

To bound the distance of $\mathcal{E}_{\mu_g}(\cdot,P)$ from the Pauli mixing distribution $\mu_{\PP^*_n}$, we consider two cases. 
We introduce a \emph{bad event}, where the random Pauli $\tilde P$ produced by the first two layers contains a row that consists entirely of identity matrices (with a corresponding good event where this does not happen).
Using the triangle inequality we decompose
\begin{equation}
\label{eq:tv_decomposition}
\TV\!\big(\mathcal{E}_{\mu_g}(\cdot,P), \mu_{\PP^*_n}\big)
\le \delta 
+ (1-\delta)\,\TV\!\big(\mathcal{E}_{\mu_R^{\otimes n_C}}(\cdot, \tilde P), \mu_{\PP^*_n}\big),
\end{equation}
where $\tilde P$ denotes the random Pauli obtained after the first two layers, conditioned on the good event. Steps 1 and 2 upper bound $\delta$ (bad event occurs), and Step~3 upper bounds the second term.

\medskip
\noindent \textbf{Step 1: First application of $\mu_R$ --- spreading across rows:} We show that after Layer~1, with high probability, at least one row has Pauli weight $\geq n_R/a$, where $a > 4/3$ is arbitrary but fixed beforehand. We define $\PP^*_{\mathrm{reg}, n_R}$ as the set of all Pauli strings with this property, and call them \emph{regular}.
It suffices to consider the worst-case input where $P$ has weight $1$.
Indeed, if $P$ has non-identity support on multiple rows, each such row spreads its weight independently under $\mu_R^{\otimes n_C}$, only increasing the chance that some row ends up with weight $\geq n_R/a$. 
Without loss of generality, we assume that $R_1 \neq \one$ and $R_c = \one$ for $c > 1$. Since $\mu_R^{\otimes n_C}$ acts independently on each row and $UR_{c}U^\dagger = \one$ whenever $R_{c} = \one$, only the first row of $P'$ after the first layer is random.
By assumption, $\mu_R$ is $\varepsilon_R$-close to Pauli mixing, which means that it approximately maps $R_1$ to a random Pauli string. We lower bound the probability that $R_1$ is mapped to a regular Pauli string.

Under the Pauli-mixing distribution $\mu_{\PP^*_{n_R}}$, the probability that an arbitrary non-identity Pauli $R \in \PP^*_{n_R}$ is mapped to an irregular Pauli is
\begin{equation}
    \sum_{R' \notin \PP^*_{\mathrm{reg}, n_R}} \mathcal{E}_{\mu_{\PP^*_{n_R}}}(R', R) = \frac{1}{4^{n_R}-1} |\{R'\in \PP^*_{n_R}: R' \notin \PP^*_{\mathrm{reg}, n_R}\}|.
\end{equation}
The number of irregular Paulis is $\sum_{k=1}^{\lceil n_R/a \rceil - 1}\binom{n_R}{k} 3^{k}$. The probability can be bounded by a tail of a binomial distribution $B(n_R, 3/4)$, which can then be bounded using Hoeffding's inequality as
\begin{equation}
     \frac{4^{n_R}}{4^{n_R}-1}\sum_{k=1}^{\lceil n_R/a \rceil - 1}\binom{n_R}{k} 3^{k}4^{-n_R} \leq \frac{4}{3}\Pr[W \leq \frac{n_R}{a}] \leq \frac{4}{3}e^{-\frac{\left(\frac{n_R}{a}-\frac{3n_R}{4}\right)^2}{n_R}}.
\end{equation}
Therefore, under the exact Pauli-mixing distribution
\begin{equation}
     \sum_{R' \in \PP^*_{\mathrm{reg}, n_R}} \mathcal{E}_{\mu_{\PP^*_{n_R}}}(R', R) \geq 1 - \frac{4}{3}e^{-\frac{9n_R}{8}\left(1 - \frac{4}{3a}\right)^2}.
\end{equation}
Since $\TV(\mathcal{E}_{\mu_R}(\cdot, Q_0), \mu_{\PP^*_{n_R}}) \leq \varepsilon_R$,
the probability of any event sampled from $\mu_R$ shifts by at most $\varepsilon_R$ with respect to the perfect Pauli-mixing\footnote{This holds for the total probability on any subset of events.}

\begin{equation}\label{eq:step1}
      \sum_{R' \in \PP^*_{\mathrm{reg}, n_R}} \mathcal{E}_{\mu_R}(R', R) \geq 1 - \left(\frac{4}{3}e^{-\frac{9n_R}{8}\left(1 - \frac{4}{3a}\right)^2} + \varepsilon_R\right).
\end{equation}

\medskip
\noindent \textbf{Step 2: Application of $\mu_C$ --- spreading across columns:} 
Define the following events for the output of the first two layers ($\mu_C^{\otimes n_R} * \mu_R^{\otimes n_C}$):
\begin{itemize}
    \item Good event: After the first two layers, every row has non-identity support. We define $\PP^*_{\mathrm{good}, n_R \times n_C}$ as all Paulis satisfying this constraint.
    \item Bad event: The complement --- at least one row consists entirely of identities.
\end{itemize}

Our goal is to lower bound the probability of the good event. To achieve this, we assume that after the initial layer, the first row is regular.
By the same reasoning as the reduction to weight-$1$ Paulis, it suffices to consider the worst case of minimal weight $\lceil n_R/a \rceil$: more non-identity entries only increase the number of columns randomized by $\mu_C$, making it easier to avoid all-identity rows. If $c$ is a column (copy of $G_C$) with a non-identity input, then $\mu_C$ is $\varepsilon_C$-close to Pauli mixing on $c$. Under the exact Pauli-mixing distribution, the probability that a given vertex $(r,c)$ receives an identity is $\frac{4^{n_C-1}-1}{4^{n_C}-1} \leq \frac{1}{4}$. Under $\mu_C$, this probability becomes at most $\frac{1}{4} + \varepsilon_C$. For a row $r$ to have the identity on every qubit after $\mu_C^{\otimes n_R}$, all $\lceil n_R/a \rceil$ affected copies of $G_C$ must independently produce identity at row $r$, giving probability at most $\left(\frac{1}{4} + \varepsilon_C\right)^{n_R/a}$. By a union bound over the $n_C$ rows, for every $P$ that has a regular row we get that
\begin{equation}\label{eq:step2}
    \sum_{P' \in \PP^*_{\mathrm{good}, n_R \times n_C}} \mathcal{E}_{\mu_C^{\otimes n_R}}(P', P) \geq 1 - n_C\left(\frac{1}{4} + \varepsilon_C\right)^{n_R/a}.
\end{equation}

\medskip
\noindent \textbf{Step 3: Second application of $\mu_R$ --- final mixing:} We now combine the previous steps to bound the full TV distance. 
From Equations~\eqref{eq:step1} and~\eqref{eq:step2}, the probability of the bad event is at most
\begin{equation}\label{eq:bad}
    \delta \coloneqq \varepsilon_R + \frac{4}{3}e^{-\frac{9n_R}{8}\left(1 - \frac{4}{3a}\right)^2} + n_C\left(\frac{1}{4} + \varepsilon_C\right)^{n_R/a}.
\end{equation}
Conditioned on the good event, every row contains at least one non-identity Pauli. The second application of $\mu_R^{\otimes n_C}$ acts independently on each of the $n_C$ rows. $\mu_R$ produces a distribution that is $\varepsilon_R$-close to uniform on $\PP^*_{n_R}$. Crucially, this $\varepsilon_R$ bound holds for \emph{any} non-identity input (it is a worst-case bound over $Q$), so the specific distribution over inputs produced by the imperfect first two layers is irrelevant. This insight allows us to decompose $\mathcal{E}_{\mu_R^{\otimes n_C} * \mu_C^{\otimes n_R} *\mu_R^{\otimes n_C}}(\cdot, Q)$ according to this event and apply the triangle inequality
\begin{equation}\label{eqn:delta_firsttwosteps}
    \max_{Q \in \PP^*}\TV(\mathcal{E}_{\mu_R^{\otimes n_C} * \mu_C^{\otimes n_R} *\mu_R^{\otimes n_C}}(\cdot, Q), \mu_{\PP^*_{n_R \cdot n_C}}) \leq \delta + (1 - \delta) \cdot \max_{Q \in \PP^*_{\mathrm{good}, n_R \times n_C}}\TV(\mathcal{E}_{\mu_R^{\otimes n_C}}(\cdot, Q), \mu_{\PP^*_{n_R \cdot n_C}}).
\end{equation}
In the bad event, the contribution to the TV distance is at most $\delta$ since this is the total probability mass of the bad branch. In the good event we upper bound the TV distance by the worst case input that could possibly occur.

To upper bound this final term we can make use of the fact that $\mu_R^{\otimes n_C}$ is acting separately on every row. Therefore, the output distribution of this layer is a product distribution over the different copies of $G_R$, and for every separate copy we know that it is $\varepsilon_R$-close to the uniform distribution. Note that the uniform distribution on all qubits $\PP^*_{n_R\cdot n_C}$ is not a product distribution, as we are omitting the all-identity Pauli. On the other hand, we have that the uniform distribution over $\PP^*_{\mathrm{good}, n_R \times n_C} = \bigtimes_{i \in [n_C]} \PP^*_{n_R}$ is a product distribution. We can easily compute that these two distributions are close to each other in TV distance
\begin{equation}\label{eqn:uniform_close_to_product}
    1 - \TV(\mu_{\PP^*_{\mathrm{good}, n_R \times n_C}}, \mu_{\PP^*_{n_R\cdot n_C})} = \frac{(4^{n_R} - 1)^{n_C}}{4^{n_R\cdot n_C} - 1} \geq 1 - \frac{n_C}{4^{n_R}}.
\end{equation}
We can use this to upper bound
\begin{equation}\label{eqn:final_TVdistance}
\begin{split}
    \max_{Q \in \PP^*_{\mathrm{good}, n_R \times n_C}}\TV(\mathcal{E}_{\mu_R^{\otimes n_C}}(\cdot, Q), \mu_{\PP^*_{n_R \cdot n_C}})
    &\leq
    \max_{Q \in \PP^*_{\mathrm{good}, n_R \times n_C}}\TV(\mathcal{E}_{\mu_R^{\otimes n_C}}(\cdot, Q), \mu_{\PP^*_{\mathrm{good}, n_R \times n_C}}) + \TV(\mu_{\PP^*_{\mathrm{good}, n_R \times n_C}}, \mu_{\PP^*_{n_R\cdot n_C})}\\
    &\leq
    \sum_{i = 0}^{n_C - 1} \max_{Q_i\in \PP^*_{n_R}}  \TV(\mathcal{E}_{\mu_R}(\cdot, Q_i), \mu_{\PP^*_{n_R}}) + \frac{n_C}{4^{n_R}}\\
    &\leq n_C \varepsilon_R + \frac{n_C}{4^{n_R}}.  
\end{split}
\end{equation}
The first step follows from the triangle inequality.
The second step follows from subadditivity, using Lemma~\ref{lem:subadditivity TV distance}, while the second term follows from Equation~\eqref{eqn:uniform_close_to_product}.
The third step follows from the fact that $\mu_R$ is $\varepsilon_R$-close in TV distance to the uniform distribution.\\

\noindent Combining Equations~\eqref{eq:bad},~\eqref{eqn:delta_firsttwosteps} and~\eqref{eqn:final_TVdistance} we can upper bound

\begin{equation}
    \max_{Q \in \PP^*_{n_R \cdot n_C}} \TV(\mathcal{E}_{\mu_g}(\cdot, Q), \mu_{\PP^*_{n_R \cdot n_C}}) \leq (n_C + 1)\,\varepsilon_R + n_C\!\left(\frac{1}{4} + \varepsilon_C\right)^{n_R/a} + \frac{4}{3}e^{-\frac{9n_R}{8}\left(1 - \frac{4}{3a}\right)^2} + \frac{n_C}{4^{n_R}},
\end{equation}
for any $a > 4/3$. For $\varepsilon_C \leq 1/4$ and $a = 3$, we can simplify to
\begin{equation}
\begin{split}
    \max_{Q \in \PP^*} \TV(\mathcal{E}_{\mu_g}(\cdot, Q), \mu_{\PP^*_{n_R \cdot n_C}})
    &\leq (n_C + 1)\,\varepsilon_R + n_C\!\left(\frac{1}{2}\right)^{n_R/3} + \frac{4}{3}e^{-\frac{9n_R}{8}\left(1 - \frac{4}{3\cdot 3}\right)^2} + \frac{n_C}{4^{n_R}}\\
    &\leq (n_C + 1)\left(\varepsilon_R + 2 e^{-n_R \log(2)/3}\right).
\end{split}
\end{equation}
\end{proof}

\noindent We now iteratively apply Lemma~\ref{lem:Cartasian_mixing} to construct 2-designs on a $D$-dimensional grid.

\begin{lemma}\label{lem:UUbar_design}
For any $D$, there exists an ensemble $\mu$ of unitaries with circuit depth $d = \order{Dn^{1/D}}$ using circuits on a $D$-dimensional grid, that forms a design with respect to experiments sampling one copy of $U$ or $U^T$ and one copy of $\bar{U}$ or $U^\dagger$, with measurable error $\varepsilon \leq 2(n + 1)e^{-n^{1/D}\log(2)/3}$.
\end{lemma}

\begin{proof}

Due to Proposition~\ref{prop:U_barU}, proving the lemma boils down to upper bounding a specific instance of the TV distance $\max_{Q\in \PP^*_n} \TV(\mathcal{E}_{\mu}(\cdot, Q), \mu_{\PP^*_n})$.
We will make use of the Cartesian-product structure of the connectivity graph present in a $D$-dimensional grid from Equation~\eqref{eqn:grid_prod_linegraphs}.
We iteratively apply Lemma~\ref{lem:Cartasian_mixing} to construct the required circuits.
Let $G_{i, \mathrm{grid}} = L_{n^{1/D}}^0 \cross \dots \cross L_{n^{1/D}}^{i}$ denote the graph after $i$ iterations of Lemma~\ref{lem:Cartasian_mixing}. 
On the 1D line graphs we sample from the Clifford group, which we know to be an exact 2-design and can be constructed in linear depth $\order{n^{1/D}}$~\cite{bravyi2021hadamard}.
We start with an exact 2-design $\mu_0$ on $G_{0, \mathrm{grid}}=L_{n^{1/D}}^0$ with depth $d_0 = \order{n^{1/D}}$.
Then for $i=1,2, \dots, D-1$, we apply Lemma~\ref{lem:Cartasian_mixing} on $G_C = G_{i, \mathrm{grid}}$ and $G_R = L_{n^{1/D}}^{i}$ to construct a 2-design $\mu_{i+1}$ on $G_{i+1, \mathrm{grid}}$ with depth $d_{i+1} = 2 \cdot \order{n^{1/D}} + d_i$ and error $2(n^{i/D} + 1)e^{-n^{1/D}\log(2)/3}$.
As long as $n^{1/D} \geq 4$, the error at each step satisfies $\varepsilon_R = 2(n^{i/D} + 1)e^{-n^{1/D}\log(2)/3} \leq \frac{1}{4}$ and we can apply the lemma.
After $D-1$ iterations, we obtain a 2-design $\mu_{D-1}$ on $n$ qubits with depth $d_{D-1} = \order{D n^{1/D}}$ and the final bound on the TV distance is
\begin{equation}
    \varepsilon \leq \max_{Q\in \PP^*_n} \TV(\mathcal{E}_{\mu}(\cdot, Q), \mu_{\PP^*_n}) \leq 2(n + 1)e^{-n^{1/D}\log(2)/3}.
\end{equation}

\end{proof} 

\begin{remark}
   Instead of taking $D$ a constant, we could choose $D = \Theta(\log(n)/\log\log(n/\varepsilon)), n^{1/D} = \Theta(\log(n/\varepsilon))$, such that we get an $n$-dependence in the depth of $d = \order{\log^2(n)}$. This matches the depth of Lemma~2 in Ref.~\cite{schuster2025strong} of strong 2-designs in the all-to-all connectivity, even though they use a different construction. Curiously, it remains an open question whether the lightcone lower bound of $\Omega(\log(n))$ on the depth of strong 2-designs in the all-to-all connectivity can be reached.
\end{remark}
We have constructed an ensemble that is secure against experiments that sample one copy of $U$ or $U^T$ and one copy of $\bar{U}$ or $U^\dagger$.
A weak relative-error 2-design is a design that is secure against experiments that sample either from $U, U^T$ or from $\bar{U}, U^\dagger$, see Lemma 6 of~\cite{schuster2025random}.
We combine the two to construct a strong design.
\begin{lemma}\label{lem:convolution_weakandstrong}
    Let $\mu_1$ be an ensemble that is secure against adversaries that query either one of $U, U^T$ and one of $\bar{U}, U^\dagger$ with measurable error $\varepsilon_1$, and let $\mu_C$ be a weak $\varepsilon_2$-approximate relative-error 2-design. Now, $\mu_1*\mu_2$ is a strong 2-design with measurable error $\max(\varepsilon_1, 2\varepsilon_2)$.
\end{lemma}
\begin{proof}
    Lemma 6 of Ref.~\cite{schuster2025random} states that a weak $\varepsilon_2$-approximate relative-error 2-design is indistinguishable from Haar random by any adversary that samples a unitary $U$ two times, up to some error $2\varepsilon_2$, but from their proof it directly follows that relative-error designs are also secure against adversaries that sample both $U$ and $U^T$. We also note that the Haar-measure is invariant under complex conjugation, so a relative-error design is also secure against adversaries that only query $\bar{U}$ and $U^\dagger$.

    Let $A$ be an experiment that queries any combination of $U, U^T, \bar{U}, U^\dagger$. Either $\mu_1$ or $\mu_2$ is secure against this experiment. If $\mu_1$ is secure and we sample $U\sim \mu_1*\mu_2$, then the outcome will be $\varepsilon_1$-close in trace distance to an experiment that queries $\mu_H * \mu_2 = \mu_H$, where $\mu_H$ is the Haar measure. An equivalent argument holds in the case that $\mu_2$ is secure against $A$.
\end{proof}
Combining our result with a construction of weak relative-error designs from Ref.~\cite{schuster2025random} we get the following theorem.

\begin{theorem}\label{theorem:2-designs_optimal}
    For any $D$ and $\varepsilon \geq 2(n+ 1)e^{-n^{1/D}\log(2)/3}$, there exists a strong $\varepsilon$-approximate unitary $2$-design using circuits on a $D$-dimensional grid on $n$ qubits, with circuit depth $d = \order{D n^{1/D} + \log(1/\varepsilon)}$.
\end{theorem}
\begin{proof}
    Ref.~\cite{schuster2025random} provides a construction of a weak relative-error unitary 2-design $\nu'$ in depth $\order{\log(n/\varepsilon)}$ in the 1D geometry. Combining this construction with Lemmas~\ref{lem:UUbar_design} and~\ref{lem:convolution_weakandstrong} proves the lemma.
\end{proof}

\begin{remark}
The weak design that is added in this theorem makes the proof easier, but is not technically needed when we apply the gluing lemma to a block of strong $k$-designs. These strong $k$-designs form a weak ``glued'' $k$-design by the gluing Lemma in~\cite{schuster2025random}, which can be used to pull out the required weak $2$-design.
On top of that, we believe that our construction in Lemma~\ref{lem:UUbar_design} is a strong design itself, since it is much deeper than known constructions of weak relative-error designs ($n^{1/D}$ instead of $\log(n)$) and has a similar structure (though the gluing lemma is not directly applicable).
\end{remark}
Now we use our strong 2-design in the strong $k$-design construction of Corollary~\ref{corollary:gluing_lemma} to prove our main result.
\begin{theorem}[Theorem~\ref{theorem:main_better}]\label{theorem:gluing_lattice_result}
For any constant $D$ and $\varepsilon \geq 6(n + 1)e^{-n^{1/D}\log(2)/3}$, there exists a strong $\varepsilon$-approximate unitary $k$-design on a $D$-dimensional grid in circuit depth $d = \order{n^{1/D} + k\log^7(k)\log(nk/\varepsilon)}$ without using auxiliary qubits.
\end{theorem}

\begin{proof}
By Theorem~\ref{theorem:2-designs_optimal} we can construct an $\varepsilon/3$-approximate strong unitary $2$-design in depth $\order{Dn^{1/D} + \log(1/\varepsilon)}$. By Lemma~\ref{lem:1d-strong-kdesigns} we can construct an $\varepsilon/n$-approximate strong unitary $k$-design on $\xi$ qubits, in depth $\order{\log^7(k)(\xi k + \log(n/\varepsilon))}$. Choosing $\xi = \frac{16}{3}\log_2(nk^2/\varepsilon) + \order{1}$, we use Corollary~\ref{corollary:gluing_lemma} to construct $\varepsilon$-approximate strong unitary $k$-designs in depth $d = \order{Dn^{1/D} + k\log^7(k)\log(nk/\varepsilon)}$.
\end{proof}

\section{Strong random unitaries in near-optimal depth through routing}\label{sec:routing}

In this section we first present a general routing procedure, and then apply it to the strong designs in all-to-all connectivity from Ref.~\cite{schuster2025strong}.
This leads to the strong $k$-designs in Theorem \ref{theorem:main_designs}, which have suboptimal scaling with the number of qubits $n$, but improved scaling in $k$ and $\varepsilon$.
While similar routing constructions have appeared in the literature~\cite{Childs_2019, Yuan_2025}, we include a self-contained derivation here to highlight a structural parallel with our construction of strong 2-designs in Section~\ref{sec:designs}.

We begin by describing the routing problem as formulated in Ref.~\cite{alon1993routing}.
Let $G(V,E)$ be a graph with vertices $V$ and edges $E$, and assume that all vertices originally contain a unique marker or ``pebble'' $p$. Now let $\pi$ be some permutation of the vertices $V$ and associate to each pebble $p$  a destination vertex $\pi(v) \in V$, such that all pebbles have distinct destination. Pebbles can be moved to different vertices of $G$ according to the following procedure: At each step a disjoint set of edges of $G$ is selected and the pebbles at the two endpoints of these edges are interchanged. Our goal is to move or ``route'' the pebbles to their respective destinations in a minimum number of steps. 

Let $p_v(t)$ denote the location of the pebble $p$ with initial location $v$ at time $t$. Thus for any $t$, the set $\{p_v(t):v\in V\}$ is just a permutation of $V$. Define $\rt(G,\pi)$ to be the minimum possible number of steps to achieve some permutation $\pi$, and define $\rt(G)$, the \emph{routing number} of $G$, as
\begin{equation}
    \rt(G) \coloneqq \max_{\pi} \rt(G,\pi).
\end{equation}

For structured graphs, we can find bounds on the routing number. For Cartesian product graphs  Ben\u{e}s~\cite{benevs1965mathematical} proved that the routing number can be bounded by the routing number of the underlying graphs, the same result was later independently proved by Baumslag and Annexstein~\cite{baumslag1991unified}:
\begin{theorem}[Theorem 1 in \cite{baumslag1991unified}] \label{thm:cartasian_routing}
For any two graphs $G$ and $G'$ we have that\footnote{Note the parallel to the depth $d_R + 2d_C$ in Lemma \ref{lem:Cartasian_mixing}.}
\begin{equation}
    \rt(G\cross G') \leq 2\, \rt(G) + \rt(G').
\end{equation}
\end{theorem}
Note that the theorem is symmetrical under exchange of $G$ and $G'$, so one can always choose the factor of two to apply to the graph with the smaller routing number. The algorithm of~\cite{baumslag1991unified} achieving Theorem~\ref{thm:cartasian_routing} works in three stages on a 2D grid: a preliminary row permutation to ensure no two pebbles in the same column share a target row, a column permutation to move each pebble to its target row, and a final row permutation to place each pebble in its correct column\footnote{Note the structural similarity to the proof of Lemma \ref{lem:Cartasian_mixing}.}. This theorem can be used to find an efficient $\SWAP$ network to permute qubits with interactivity constraints given by a $D$-dimensional grid.

\begin{lemma}[Efficient permutations in $D$ dimensions]
\label{lem:route_grid_unitary}
Given $n$ qubits distributed over a $D$-dimensional grid with nearest-neighbor interactions, then for any permutation $\pi$ on the set of qubits there exists a unitary $U_\pi$ implementing this permutation in depth $d(U_\pi) = \order{D n^{1/D}}$.
\end{lemma}
\begin{proof}
As described above, any routing scheme can be implemented by a $\swap$ circuit with depth equivalent to the number of steps required by the routing protocol. Let $U_\pi$ be this $\swap$ circuit, by definition
\begin{equation}
    d(U_\pi) \leq \rt(G_{D,\mathrm{grid}}).
\end{equation}
Therefore, it is sufficient to upper bound the routing time of $G_{D, \mathrm{grid}}$. The routing number of a line graph of $\ell$ vertices is $\rt(L_\ell) = \order{\ell}$. By definition, a $D$-dimensional grid is the Cartesian product of $D$ line graphs. Iteratively taking the Cartesian product of line graphs, applying Theorem~\ref{thm:cartasian_routing}, and making sure that the factor $2$ is applied to the line graph that is being added, we get
\begin{equation}
    d(U_{\pi}) \leq \rt(G_{D,\mathrm{grid}}) \leq 2 \rt(L_{n^{1/D}}) + \rt(G_{D-1,\mathrm{grid}}) \leq 2D\cdot \rt(L_{n^{1/D}}) \leq \order{Dn^{1/D}}.
\end{equation}
\end{proof}

\noindent This lemma allows us to bound the depth of compiling all-to-all circuits onto a $D$-dimensional grid:
\begin{theorem}
\label{thm:routing}
Let $V$ be a quantum circuit acting on $n$ qubits, consisting of $2$-local gates with all-to-all connectivity and of depth $d(V)$. Given access to a  $D$-dimensional grid containing $n$ qubits, then there exists a unitary $U_{\mathrm{grid}}$, with circuit depth $d(U_{\mathrm{grid}}) = \order{d(V)Dn^{1/D}}$ that implements $V$ acting on the qubits in the grid and respects grid topology.    
\end{theorem}

\begin{proof}
We can first decompose $V$ into $d(V)$ layers, $V = V_{d(V)-1} \dots V_1 V_0$. These layers can be further decomposed into  $2$-local gates acting on a disjoint set of pairs of qubits, $V_i = \prod_{(l,k) \in \mathsf{I}_{i}} \tilde{V}_{l,k}$, where $\tilde{V}_{l,k}$ only acts on qubits $l$ and $k$ and $\mathsf{I}_{i}$ is the set of interactions, i.e., a matching of the $n$ qubits.
Note that these interaction might be highly non-local, and for these gates to be applied within the grid topology the qubits $l$ and $k$ have to be placed next to each other.
Let $\pi_i$ be a permutation such that all the interactions $\mathsf{I}_i$ become local (note that $\pi_i$ is not uniquely defined and might depend on the choice of $\pi_j$ for $j<i$). Then the following circuit implements $V$ on a $D$-dimensional grid:
\begin{equation}
    U_{\mathrm{grid}} = U_{\pi_{d(V)}}V_{d(V)-1}U_{\pi_{(d(V)-1)}}\dots V_1 U_{\pi_1} V_0 U_{\pi_0}.
\end{equation}
The last permutation $U_{\pi_{d(V)}}$ is required to reset the qubits into their original order. As $d(V_i) = \order{1}$ and $d(U_{\pi_i}) = \order{Dn^{1/D}}$ by Theorem~\ref{lem:route_grid_unitary}, we conclude that $d(U_{\mathrm{grid}}) = \order{d(V) D n^{1/D}}$.
\end{proof}

\noindent We now apply this routing protocol to the all-to-all random circuits in Theorem~\ref{thm:all-to-all-designs}, resulting in near-optimal approximate designs in $D$ dimensions:

\begin{corollary}[Theorem~\ref{theorem:main_designs}]
\label{cor:routed_result}
For any constant $D$, there exist strong $\varepsilon$-approximate unitary $k$-designs in the following circuit depth using circuits on a $D$-dimensional grid:
\begin{enumerate}
        \item $d = \order{(nk)^{1/D}\poly\log(n,k) \cdot \log\log(1/\varepsilon)}$ with $\torder{nk}$ auxiliary qubits.
        \item $d = \order{n^{1/D} k \cdot \poly\log(n) \cdot \log\log(k/\varepsilon)}$ with $\torder{n}$ auxiliary qubits.
\end{enumerate}
\end{corollary}
Note that the first result gets an additional scaling of $k^{1/D}$ due to the fact that the circuit acts on $\torder{nk}$ qubits, therefore routing has to be performed on a larger grid. This results in sublinear scaling with respect to $k$ at the cost of more auxillary qubits. 

\section{Optimal-depth strong PRUs} \label{sec:pru}

In this section we construct strong PRUs on a $D$-dimensional grid using the procedure in Corollary~\ref{corollary:PRU_gluing}.
We use the strong 2-design we developed in Section~\ref{sec:designs}, and we use the routing results from Section~\ref{sec:routing} to construct the small PRUs.

\begin{theorem}[Theorem~\ref{theorem:main_pru}]
\label{theorem:pru}
    Assuming sub-exponential post-quantum hardness of LWE, on a $D$-dimensional grid for constant $D$, one can construct strong PRUs in depth $\order{n^{1/D}}$ without auxiliary qubits.
\end{theorem}
\begin{proof}
    By Lemma~\ref{lemma:PRU_buildingblocks} one can construct a strong PRU in the all-to-all connectivity in depth $d = \order{n^\ell}$ for some constant $\ell$, and this PRU is secure against attacks running in time $2^{\order{n^\delta}}$, with $\delta>0$ some constant. We can use this to construct strong PRUs on $\xi = \log^{2/\delta}(n)$ qubits in depth $\order{(\log^{2/\delta}(n))^\ell} = \polylog(n)$. 
    These PRUs are secure against attacks running in time $2^{\order{\log^2(n)}}$, which is larger than any polynomial in $n$.
    We can now use routing, Theorem~\ref{thm:routing}, to construct this PRU in depth $\order{D\log^{2/(\delta D)}(n) \cdot \polylog(n)} = \polylog(n)$. We finally combine Corollary~\ref{corollary:PRU_gluing} and Theorem~\ref{theorem:2-designs_optimal} to prove the theorem.
\end{proof}
\begin{acknowledgements}
We would like to thank Thomas Schuster for clarifications on the strong gluing lemma and Jonas Haferkamp for useful conversations. JH acknowledges funding from the Dutch Research Council (NWO) through a Veni grant (grant No.VI.Veni.222.331) and the Quantum Software Consortium (NWO Zwaartekracht Grant No.024.003.037). AD is funded through an Ada Lovelace fellowship from the Quantum Software Consortium. LG is funded by NWO through NGF.1623.23.005. MF acknowledges funding from the Quantum Software Consortium.

\end{acknowledgements}

\bibliographystyle{alphaurl}
\bibliography{bibliography}

\end{document}